\documentclass[11pt]{article}
\usepackage{amssymb}
\usepackage{amsfonts,amsmath, longtable}

\usepackage{comment}

\newtheorem{theorem}{Theorem}

\newcommand{\tr}{{\rm tr}}

\newcommand{\om}{\omega}

\newcommand{\de}{\delta}

\newcommand{\Mat}{ {\rm Mat}_N }

\newcommand{\mC}{\mathbb C}
\newcommand{\mZ}{\mathbb Z}

\newcommand{\rank}{\mathop{\rm rank}}
\newcommand{\h}{\hbar}

\newcommand{\mn}{\mathrm{Mat}_{\infty}\big(\mathbb{C}((\h))\big)}

\newtheorem{lemma}{Lemma}[section]

\newenvironment{proof}{\par\noindent{\bf Proof.}}{\hfill$\scriptstyle\blacksquare$}

\def\beq{\begin{equation}}
\def\eq{\end{equation}}

\def\res{\mathop{\hbox{Res}}\limits}

\begin{document}

\setcounter{page}{1}

\begin{center}

\


{\Large{\bf Large $N$ limit of spectral duality between }}

\vspace{3mm}

{\Large{\bf the classical XXX spin chain and the rational reduced Gaudin model}}

\vspace{3mm}



 \vspace{12mm}

 {\Large {R. Potapov}}
\qquad\quad\quad

  \vspace{10mm}

 {\em Steklov Mathematical Institute of Russian
Academy of Sciences,\\ Gubkina str. 8, 119991, Moscow, Russia}

{\em Institute for Theoretical and Mathematical Physics,\\
Lomonosov Moscow State University, Moscow, 119991, Russia}
%


 {\small\rm {e-mail: trikotash@ya.ru}}

\end{center}

\vspace{0mm}

\begin{abstract}
We study the large $N$ limit of the spectral duality between the classical $\mathfrak{gl}_M$ XXX spin chain and the $\mathfrak{gl}_N$ trigonometric Gaudin model in its rational reduced form. 
The infinite-dimensional limit of the Gaudin model is constructed within the framework of the noncommutative torus algebra, following the approach of Hoppe, Olshanetsky and Theisen. 
In this formulation, the matrix dynamical variables are promoted to fields on the torus, and the corresponding Lax equations acquire the Moyal star-product structure. 
The dual model is obtained as the large $N$ limit of the $\mathfrak{gl}_M$ classical XXX spin chain with $N$ sites, represented by Laurent series satisfying a quadratic $r$-matrix relation associated with the rational solution of the classical Yang--Baxter equation. 
\end{abstract}

%

{\small{ \tableofcontents }}

\bigskip\

\section{Introduction}\label{sec1}
\setcounter{equation}{0}
In this paper we continue to study large $N$ limits of spectral dualities in classical integrable models of Gaudin type \cite{G}. We consider the infinite-dimensional limit of Gaudin models represented by the algebra of noncommutative torus ($A_{\h}$) \cite{HPS,BHS,JS, R}. In this approach, the dynamical matrix variables become fields on the torus with the Lax equation
\beq\label{i1}
\displaystyle{
\partial_{t}L(\vec{\phi}, z) = M(\vec{\phi},z)\star L(\vec{\phi},z) - L(\vec{\phi},z)\star M(\vec{\phi},z)\,,
}
\eq
where $L(\vec{\phi},z)$ and $M(\vec{\phi},z)$ are functions on $\mathbb{T}^{2}\times \mathbb{R}$  and $\star$ denotes the Moyal star product \cite{Weyl}. $A_{\h}$ possesses an invariant trace 
\beq\label{ii2}
\displaystyle{
	\tr f = \frac{1}{4\pi^{2}}\int_{\mathbb{T}^{2}}\mathrm{d}\vec{\phi}\,f(\vec{\phi})\,.
}
\eq
which makes traces of powers of $L(z,\vec{\phi})$ conserved quantities. Thus, we obtain a noncommutative field theory \cite{DN}, moreover it appears to be Hamiltonian with infinitely many conserved charges in involution. This limit was introduced by Hoppe, Olshanetsky and Theisen and was applied to various integrable systems \cite{HPS, HOT,Olsh}. Here, we focus on its application to the spectral duality of Gaudin-type models. 

Recall that two finite-dimensional models with spectral parameter dependent Lax matrices $L(z)$ and $\tilde{L}(v)$ are spectrally dual if their spectral curves  
\beq\label{i2}
\displaystyle{
	\det(v - L(z)) = 0 \,,
}
\eq
\beq\label{i3}
\displaystyle{
	\det(z - \tilde{L}(v)) = 0 \,,
}
\eq
coincide.
This fact was first observed for the Toda chain \cite{FT} and for the pair of $\mathfrak{gl}_{N}$ and $\mathfrak{gl}_{M}$ rational Gaudin models \cite{AHH}. Later, the same correspondence was established for the classical XXX spin chain and trigonometric Gaudin models \cite{MMZZ,MMZZR1}, as well as for the classical XXZ spin chains \cite{MMZZR2}. The duality has numerous applications. For instance, it appears in the analysis of gauge theories and the corresponding Seiberg-Witten geometry \cite{MMZZ, MM, N1}. Furthermore, it may be generalized to the cases of monodromy preserving equations, the Knizhnik-Zamolodchikov equation \cite{GVZ} and quantum integrable models \cite{MTV, VY}. For a detailed discussion of dualities in finite-dimensional integrable systems see \cite{GVZ, PZ} and references therein. The inifnite-dimensional spectral duality was discussed in the theory of matrix models \cite{BEHL}, here we use the noncommutative torus approach. 

In our previous work \cite{PZ1}, we obtained a field-theoretical version of the spectral duality in rational Gaudin models. For systems described by the noncommutative torus, the role of the spectral curve is played by the spectral power series in $v^{-1}$:
\beq\label{i4}
\displaystyle{
	\Gamma_{\infty}(v,z) = \exp \tr \ln_{\star}(I - \frac{1}{v}L(z,\phi))\,.
}
\eq
It appears that for a special form of the $A_{\h}$ Gaudin models one may construct a dual model with a finite-dimensional Lax matrix corresponding to the large N limit of Gaudin models with irregular singularities \cite{VY,FFRL}, such that their spectral power series are expressed through one another. 

Here we study a more complicated case of the classical XXX spin chain and the trigonometric Gaudin model. The phase space of the first model is $\mathfrak{gl}_{M}^{*\times N}$. It is defined by the classical monodromy matrix: 
\beq\label{i5}
\displaystyle{
	T(v) = Z(I+\frac{S^{N}}{v-v_{N}})...(I + \frac{S^{1}}{v-v_{1}}) \in \mathrm{Mat}_{M}\,,
}
\eq
with $Z = \mathrm{diag}\{z_{1},...,z_{M}\}$. The duality holds for the $\mathrm{rank} = 1$ parametrization \footnote{Representation \eqref{i6} comes from \eqref{i5} by brackets expanding and setting $v_{i} \neq v_{j}$.}  
\beq\label{i6}
\displaystyle{
	T_{ab}(v) = z_{a}\delta_{ab} + \sum\limits_{k = 1}^{N}\frac{z_{a}\xi^{a}_{k}\eta^{b}_{k}}{v-v_{k}}\, \quad a,b = 1,...M\,.
}
\eq
The same phase space also carries the rational reduced Gaudin model
\beq\label{i7}
\displaystyle{
	L_{ij}(z) = v_{i}\delta_{ij} + \sum\limits_{k = 1}^{M}\frac{z_{k}\xi^{k}_{i}\eta^{k}_{j}}{z-z_{k}}\,,\quad i,j = 1,...,N\,,
}
\eq
which is gauge equivalent to the trigonometric Gaudin model \cite{MMZZR2}. In \cite{MMZZR1} it was shown that their spectral curves coincide.

\paragraph{Purpose of the paper.} The main goal is to extend the large $N$ limit approach, previously applied to pairs of rational Gaudin models \cite{PZ1}, to the case of the classical XXX spin chain and the trigonometric Gaudin model. Naively speaking, the $N \to \infty$ limit of the monodromy matrix \eqref{i5} should be the monodromy matrix of Heisenberg magnet field theory \cite{FT}. In this limit, the dual model Lax matrix \eqref{i7} becomes an infinite-dimensional matrix. For this reason, we perform the $A_{\h}$ limit of the Gaudin model \eqref{i7}. Note that, for the purposes of the spectral duality, both models $(\ref{i6}) - (\ref{i7})$ are written in sophisticated forms, making Poisson brackets of their dynamical variables are rather complicated. This form of the dynamical variables is the main obstruction to generalizing the $A_{\h}$ spectral duality for such models.

Our first goal is to define the $A_{\h}$ version of the rational reduced model \eqref{i7}. For $\mathfrak{gl}_{N}$, it is obtained from the untwisted rational model \eqref{i3} via reduction by the coadjoint action of $GL_{N}$. While the same procedure can be performed for infinite matrices, the calculation of the reduced Poisson structure involves inverse matrices that cannot be represented by functions on the torus. It was also shown that the finite-dimensional model is gauge equivalent to the trigonometric Gaudin model \cite{MMZZR2}. The $A_{\h}$ limit of the trigonometric model is a straightforward generalization of the $A_{\h}$ limit of the rational one. However, the gauge transformation does not lie in any algebra of infinite-dimensional matrices and also does not correspond to any well-defined function on the torus. Nevertheless, under suitable convergence conditions the $A_{\h}$ reduced model can be described for the $\rank = 1$ case.

The second goal is to describe the infinite-dimensional limit of the dual XXX spin chain \eqref{i6}. For the homogeneous case ($v_{i} = 0$) this limit leads to the 1+1 soliton Heisenberg Magnet equation. Instead of this, we keep all $v_{i} \neq v_{j}$ and obtain a monodromy matrix with higher order poles in the spectral parameter. We show that it satisfies quadratic the $r$-matrix relation for the rational solution of the classical Yang-Baxter equation.

The paper is organized as follows. In the first section we recall the definition and properties of the noncommutative torus algebra and the $A_{\h}$ limit of rational Gaudin model. In Section 2, we review the definitions of finite-dimensional models and their spectral duality. Finally, in Section 3 we define the $A_{h}$ models and their duality.

\section{$A_{\h}$ algebra and infinite-dimensional limit of Gaudin model}
\setcounter{equation}{0}
In this section, we recall the definitions and basic properties of the $A_{\h}$ algebra and its relation to integrable field theories. Here, we provide only the facts needed for the infinite-dimensional limits. Further details and references can be found in \cite{R, HOT,Olsh,PZ1}.
\subsection{Algebra of noncommutative torus and its representations}
We define the noncommutative torus algebra $A_{\h}$ as the algebra of infinite series in two noncommutative variables, with coefficients in the ring of formal Laurent series $\mathbb{C}((\h))$:

	\beq\label{n1}
	\displaystyle{
		A_{\h} = \mathbb{C}((\h))\langle\langle U_{1},U_{2}\rangle\rangle/\langle U_{1}U_{2} - \omega U_{2}U_{1}\rangle, \,\,\,\,\, \omega = e^{4\pi i \hbar}\,.
	}
	\eq
It is convenient to introduce the basis:
\beq\label{n2}
\displaystyle{
	 T_{\vec{m}} = \frac{i}{4\pi\hbar} \omega^{\frac{m_{1}m_{2}}{2}} U_{1}^{m_{1}}U_{2}^{m_{2}}\,, \quad \vec{m} \in \mZ^{2}\,,
}
\eq
such that any element of $A_{\h}$ can be represented as
\beq\label{nn1}
\displaystyle{
	S = \sum\limits_{\vec{m} \in \mZ^{2}} S_{\vec{m}}T_{\vec{m}}\,.
}
\eq
The multiplication law now reads  
\beq\label{n3}
\displaystyle{T_{\vec{m}}T_{\vec{n}}  = \frac{i}{4 \pi \hbar}
	\omega^{-\frac{1}{2}\vec{m}\times\vec{n}}T_{\vec{m}+\vec{n}}\,,	\qquad \vec{m} \times \vec{n} = m_{1}n_{2} - n_{1}m_{2}\,.
}
\eq

As an associative algebra $A_{\h}$, has two useful representations. The first one is in $\mn$, the algebra of matrices with finitely many non-zero diagonals over the ring $\mC((\h))$, via
\beq\label{n4}
\displaystyle{
	T_{\vec{m}} = \sum\limits_{i,j \in \mZ}W_{\vec{m},ij}E_{ij}\,,
}
\eq
where $E_{ij}$ is a standard basis of matrix unities and
\beq\label{n5}
\displaystyle{W_{\vec{m},ij} = \frac{i}{4\pi\h}w^{\frac{m_{1}m_{2}}{2}}w^{m_{1}i}\delta_{j,i+m_{2}}\,,\quad \vec{m} \in \mZ^{2}\,,\quad i,j \in \mZ\,.	
}
\eq
Representation \eqref{n4} is, in fact, an embedding, which makes $A_{\h}$ a subalgebra of $\mn$.  For matrices in this subalgebra, one can invert \eqref{n5}, such that
\beq\label{n6}
\displaystyle{
	E_{ij} = \sum\limits_{\vec{m} \in \mZ^{2}}W^{-1}_{ij,\vec{m}}T_{\vec{m}}\,,
}
\eq
\beq\label{n7}
\displaystyle{
	W_{ij,\vec{n}}^{-1} = \frac{4 \pi \h}{i}w^{-\frac{i+j}{2}n_{1}}\delta_{n_{2},j-i}\,,\quad \vec{n} \in \mZ^{2}\,,\quad i,j \in \mZ\,.	
}
\eq
In what follows, we call the coordinates of $A_{\h}$ elements in the $\mn$ representation their matrix modes.

The second representation takes values in $C^{\infty}(\mathbb{T}^{2})((\h))$ with the Moyal star product:
\beq\label{nn0}
\displaystyle{
	f\star g = fg + \sum\limits_{k = 1}^{\infty}\frac{i^{k}\h^{k}}{2^{k}k!}\sum\limits_{i_{1},j_{1},...i_{k},j_{k} = 1}^{2}\pi^{i_{1}j_{1}}...\pi^{i_{k}j_{k}}\partial_{i_{1}}...\partial_{i_{k}}f\partial_{j_{1}}...\partial_{j_{k}}g\,,
}
\eq
where 
\beq\label{nnn1}
\displaystyle{
	\pi = \sum\limits_{i,j = 1,2} \pi^{ij}\partial_{i}\wedge \partial_{j} =  \frac{\partial}{\partial \phi_{1}}\wedge \frac{\partial}{\partial \phi_{2}}\,,
}
\eq
is a Poisson bivector on the torus with coordinates $\phi_{1,2} \in (-\pi,\pi)$.
Consider a function on the torus with Fourier expansion
\beq\label{nn2}
\displaystyle{
	 f(\vec{\phi}) = \sum\limits_{\vec{m}\in \mZ^{2}}f_{\vec{m}}e^{i \vec{m} \vec{\phi}}\,.
}
\eq
The Moyal star product \eqref{nn0} provides noncommutative associative multiplication in $C^{\infty}(\mathbb{T}^{2})((\h))$, which for basis elements is written as
\beq\label{nn3}
\displaystyle{
	e^{i \vec{m} \vec{\phi}} \star 	e^{i \vec{n} \vec{\phi}} = \om^{-\frac{1}{2}\, \vec{m} \times \vec{n}}e^{i (\vec{m}+\vec{n})\vec{\phi}}\,.
}
\eq
Thus, the representation of $A_{\h}$ is obtained by the identification
\beq\label{n8}
\displaystyle{T_{\vec{m}} = \frac{i}{4 \pi \h}\,e^{i \vec{m} \vec{\phi}}\,.
}
\eq
This definition may be generalized to the function from $L^{1}(\mathbb{T}^{2})$ by writing the star product in integral form
\beq\label{nn4}
\displaystyle{
	f \star g \,(\vec{\phi}) = - \frac{i}{\pi^{2}(4 \pi \h)^{3}} \int_{\mathbb{T}^{2} \times \mathbb{T}^{2}}\mathrm{d} \vec{\phi'} \mathrm{d} \vec{\phi''} e^{\frac{2i}{4 \pi \h}|\vec{\phi}\vec{\phi'}\vec{\phi''}|}f(\vec{\phi}')g(\vec{\phi''})\,.
}
\eq
The crucial property for integrability is the existence of a trace in $A_{\h}$ with the invariance property
\beq\label{n10}
\displaystyle{
	\tr AB = \tr BA\,, \quad A,B \in A_{\h}\,.
}
\eq
For basis elements \eqref{n2} it is equal
\beq\label{n9}
\displaystyle{
	\tr T_{m_{1}m_{2}} = \frac{i}{4 \pi \h}\delta_{m_{1},0}\delta_{m_{2},0}\,,
}
\eq
Note that in $(C^{\infty}(\mathbb{T}^{2})((\h)),\star)$ representation it takes form
\beq\label{n11}
\displaystyle{
	\tr f = \frac{1}{4 \pi^{2}}\int \limits_{\mathbb{T}^{2}}\mathrm{d}\vec{\phi} f(\vec{\phi})\,.
}
\eq
\subsection{The Poisson-Lie structure}
As a Lie algebra, $A_{\h}$ coincides with sin-algebra \cite{FFZ}:
\beq\label{n13}
\displaystyle{ [T_{\vec{m}},T_{\vec{n}}] = \frac{1}{2 \pi \h} \sin (2 \pi \h \,\vec{m} \times \vec{n})T_{\vec{m}+\vec{n}}\,.
}
\eq
In this form $A_{\h}$ appears as the $N \to \infty$ limit of $\mathfrak{gl}_{N}$ \cite{HPS,BHS}.
We endow the dual Lie algebra $L^{*}_{\h}$ with a Poisson-Lie bracket
\beq\label{n14}
\displaystyle{
	\{S_{\vec{m}},S_{\vec{n}} \} =  -\frac{1}{2 \pi \h} \sin (2 \pi \h \,\vec{m} \times \vec{n})S_{\vec{m}+\vec{n}}\,,
}
\eq
where $S_{\vec{m}}$ are coordinates on $L^{*}_{\h}$, such that for a generic element one has 
\beq\label{n15}
\displaystyle{
	S = \sum\limits_{\vec{m} \in \mZ^{2}}S_{\vec{m}}T_{\vec{m}}\,.
}
\eq
The minus sign is chosen for convenience in the $r$-matrix definition of the Gaudin model. Similarly to the $\mathfrak{gl}^{*}_{N}$ case, this bracket may be rewritten in the r-matrix form
\beq\label{n16}
\displaystyle{
	\{S_{1},S_{2}\} = -\left[\tilde{P}_{12},S_{1}\right]\,,
}
\eq
where
\beq\label{n17}
\displaystyle{
	\tilde{P}_{12} = \sum\limits_{\vec{m} \in \mZ^{2}} T_{\vec{m}}\otimes T_{-\vec{m}}\, \in A_{\h}^{\otimes 2}\,.
}
\eq
Finally, for a field on the torus 
\beq\label{n19}
\displaystyle{
	S(\vec{\phi}) = \sum\limits_{\vec{m} \in \mZ^{2}}S_{\vec{m}}e^{i\vec{m}\vec{\phi}}\,,
}
\eq
the bracket \eqref{n16} provides 
\beq\label{n20}
\displaystyle{
	\{S(\vec{\phi}), S(\vec{\theta})\} = -\frac{i}{4 \pi \h}[S^{a}(\vec{\phi}),\delta(\vec{\theta}-\vec{\phi})]_{\star}\,.
}
\eq
\subsection{Infinite-dimensional Gaudin model}
The Poisson-Lie bracket \eqref{n13} allows one to construct the large N limit of Gaudin models \cite{Olsh, PZ1}. Here, we briefly recall this limit for the rational case. The phase space of the $A_{\h}$ rational Gaudin model is $L_{\h}^{* \times M}$ with the direct sum of Poisson-Lie structures: 
\beq\label{g1}
\displaystyle{
		\{S_{\vec{m}}^{a},S_{\vec{n}}^{b}\} = -\frac{\delta^{ab}}{2 \pi \h} \sin (2 \pi \h \, \vec{m} \times \vec{n})S_{\vec{m}+\vec{n}}^{a}\,.
}
\eq
The model is defined by the Lax matrix:
\beq\label{g2}
\displaystyle{
	L(z) = \sum\limits_{m \in \mZ} v_{m}T_{m0}+   \sum\limits_{k = 1}^{M} \sum\limits_{\vec{m} \in \mZ}\frac{S^{k}_{\vec{m}}}{z-z_{k}}T_{\vec{m}} \in A_{\h}\,,
}
\eq
with constants $v_{i}$ and $z \in \mathbb {CP}^{1}\backslash\{z_1,...,z_M\}$ is the spectral parameter.
Integrability follows from the linear $r$-matrix structure
\beq\label{g3}
\displaystyle{
	\{L_{1}(z),L_{2}(w)\} = [r_{12}(z-w), L_{1}(z)+L_{2}(w)]\,,
}
\eq
with 
\beq\label{g4}
\displaystyle{
	r_{12}(z) = \frac{1}{z}\sum\limits_{\vec{m} \in \mZ^{2}} T_{\vec{m}}\otimes T_{-\vec{m}} \in A_{\h}^{\times 2}\,.
}
\eq
This structure guarantees that the traces of powers of \eqref{g2} are Poisson commuting functions:
\beq\label{g5}
\displaystyle{
	\{\tr L^{a}(z), \tr L^{b}(w)\} = 0\, \quad \forall a,b \in \mathbb{N}\,.
}
\eq
Note that in the $A_{\h}$ case, the Lax matrix \eqref{g2} generates infinitely many independent conserved charges. Thus, one may construct an integrable field theory using the $(C^{\infty}(\mathbb{T}^{2}((\h)),\star)$ representation. After rescaling by $\frac{i}{4\pi \h}$, \eqref{g3} becomes
\beq\label{g6}
\displaystyle{
	L(z,\vec{\phi}) = V(\phi_{1}) + \sum\limits_{k = 1}^{M}\frac{S^{k}(\vec{\phi})}{z-z_{k}}\,,
}
\eq
with $M$ fields on $\mathbb{T}^{2} \times \mathbb{R}$
\beq\label{g7}
\displaystyle{
	S^{k}(\vec{\phi},t) = \sum\limits_{\vec{m} \in \mZ}S^{k}_{\vec{m}}(t)e^{i\vec{m}\vec{\phi}}\, 
}
\eq
and a constant "twist function" $V(\phi_{1})$.  This makes the $A_{\h}$ Gaudin model a Hamiltonian field theory with the Poisson structure
\beq\label{g8}
\displaystyle{
	\{S^{a}(\vec{\phi}), S^{b}(\vec{\theta})\} = -\delta^{ab}\frac{i}{4 \pi \h}[S^{a}(\vec{\phi}),\delta(\vec{\theta}-\vec{\phi})]_{\star}\,,
}
\eq
and the Hamiltonian function
\beq\label{g9}
\displaystyle{
	H_{i} = \frac{1}{2}\,\res_{z = z_{i}}\tr L^{2}(z)\, = \frac{1}{4 \pi^{2   }}\int_{\mathbb{T}^{2}}\mathrm{d}\vec{\phi}\,V(\phi_{1})\star S^{i}(\vec{\phi})+\frac{1}{4 \pi^{2   }}\sum\limits_{k:k \neq i}^M\int_{\mathbb{T}^{2}}\mathrm{d}\vec{\phi}\,\frac{S^{i}(\vec{\phi}) \star S^{k}(\vec{\phi})}{z_{i}-z_{k}}\,.
}
\eq

Equations of motion for \eqref{g9} take the Lax form 
\beq\label{g10}
\displaystyle{
	\partial_{t_{i}}L(z,\vec{\phi}) = \frac{i}{4 \pi \h}[M_i(z,\phi),L(z,\phi)]_{\star}\,,
}
\eq
for the Lax function \eqref{g6} and
\beq\label{g11}
\displaystyle{
	M_{i}(z,\vec{\phi}) = -\frac{S^{i}(\vec{\phi})}{z-z_{i}}\,.
}
\eq
The property \eqref{g5} makes \eqref{g10} an integrable field theory possessing infinitely many commuting conserved charges generated by
\beq\label{g12}
\displaystyle{
	Q^{k}(z) = \tr L^{\star k}(z) = \frac{1}{4\pi^{2}}\int_{\mathbb{T}^{2}}\mathrm{d}\vec{\theta}\,L^{\star k}(z,\vec{\theta})\,.
}
\eq
\section{Spectral duality of finite-dimensional models}
\setcounter{equation}{0}
In this section, we recall the duality statement for the $\mathfrak{gl}_{N}$ rational reduced Gaudin model and the $\mathfrak{gl}_{M}$ XXX classical spin chain.
\subsection{$\mathfrak{gl}_{N}$ rational reduced Gaudin model}
We begin with the direct sum of Poisson-Lie structures on $\mathfrak{gl}_{N}^{*\times M+2}$: 
\beq\label{fr1}
\displaystyle{
	\{S^{a}_{ij},S^{b}_{kl}\} = \delta_{ab}(S^{a}_{il}\delta_{kj} - S^{a}_{kj}\delta_{il})\,.
}
\eq
This space admits the Poisson action of $\mathrm{GL}_{N}$, represented by conjugation
\beq\label{fr2}
\displaystyle{
	S^{k} \mapsto gS^{k}g^{-1}\,\quad k = 1,...,M+2\,.
}
\eq
Therefore, one may perform a Poisson reduction with respect to this action, which yields the following first class constraints 
\beq\label{fg2}
\displaystyle{
	\sum\limits_{k = 1}^{M+2} S^{k} = 0\,.
}
\eq

Let us now show how this procedure affects the untwisted rational Gaudin model. Recall that on $\mathfrak{gl}_{N}^{*\times M+2}$ this model is defined by the Lax matrix
\beq\label{fg3}
\displaystyle{
	L(z) = \frac{S^{0}}{z} + \sum\limits_{k =1}^{M}\frac{S^{k}}{z-z_{k}}\,.
}
\eq
and the $r$-matrix relation
\beq\label{fgg1}
\displaystyle{
	\{L_{1}(z),L_{2}(w)\} = \Big[\frac{P_{12}}{z-w}, L_{1}(z) + L_{2}(w)\Big]\,.
}
\eq
for the permutation matrix
\beq\label{fgg2}
\displaystyle{
	P_{12} = \sum\limits_{i,j = 1}^{N}E_{ij}\otimes E_{ji}\in \Mat\,.
}
\eq
Here, we fix two marked points to be $0$ and $\infty$,  due to the three-transitive action of $\mathrm{SL}_{2}$ on $\mathbb{CP}^{1}$. The coadjoint  action \eqref{fr2} gives rise to the moment map
\beq\label{fg4}
\displaystyle{
	\mu_{ij} = \sum\limits_{k = 1}^{M+2}S^{k}_{ij}\,.
}
\eq
Setting $\mu = 0$, we arrive at the set of first class constraints \eqref{fg2}. One way to deal with first class constraints systems is to impose additional constraints, commonly referred as gauge fixing. To obtain the model that appears in the context of spectral duality, one needs to choose 
\beq\label{fg5}
\displaystyle{
	\mu_{ij} = 0 \,,\quad i\neq j\,,\quad i,j = 1,...,N\,,
}
\eq
\beq\label{fg6}
\displaystyle{
	S^{\infty}_{ij} = 0 \,,\quad i\neq j\,,\quad i,j = 1,...,N\,.
}
\eq
Taken together, $(\ref{fg5})-(\ref{fg6})$ are constraints of the second class and may be solved using the Dirac bracket. For $a,b \neq 0,\infty$ it is equal 
\beq\label{fg7}
\begin{array}{c}
	\displaystyle{
		\{S^{a}_{ij},S^{b}_{kl}\} = \delta^{ab}\left(S_{il}^{a}\delta_{kj} - S_{kj}^{a}\delta_{il} \right) -
	}
	\\ \ \\
	\displaystyle{
		- \sum\limits_{p \neq k}\frac{S^{a}_{ip}S^{b}_{pl}\de_{jk}}{v_{p}-v_{k}}	- \sum\limits_{p \neq i}\frac{S^{b}_{kp}S^{a}_{pj}\de_{il}}{v_{i}-v_{p}}+ 	 \frac{S^{a}_{il}S^{b}_{kj}(1-\de_{jl})}{v_{l}-v_{j}} + \frac{S^{a}_{kj}S^{b}_{il}(1-\de_{ik})}{v_{k}-v_{i}}\,.
	}
\end{array}
\eq
where $S^{\infty}_{ii} = -v_{i}$.

For $\mathrm{rank} = 1$ variables  $S^{a}_{ij} = \xi^{a}_{i}\eta^{a}_{j}$ this structure follows from \cite{MMZZR1}:  
\beq\label{fg8}
\displaystyle{
	\left\{\xi_{i}^{a},\eta_{j}^{b}\right\} = -\delta_{ij}\left(\delta^{ab}+\sum \limits_{p \neq i }\frac{\xi_{p}^{a}\eta_{p}^{b}}{v_{i}-v_{p}}\right)\,,	
}
\eq
\beq\label{fg9}
\displaystyle{
	\left\{\xi_{i}^{a},\xi_{j}^{b}\right\} = \frac{\xi_{j}^{a}\xi_{i}^{b}\left(1 - \delta_{ij}\right)}{v_{i}-v_{j}}\,,
}
\eq
\beq\label{fg10}
\displaystyle{
	\left\{\eta_{i}^{a},\eta_{j}^{b}\right\} = \frac{\eta_{j}^{a}\eta_{i}^{b}\left(1 - \delta_{ij}\right)}{v_{j}-v_{i}}\,.	
}
\eq

The final step is to solve \eqref{fg2} with respect to $S^0$
\beq\label{fg11}
\displaystyle{
	S^{0} = -S^{\infty} - \sum\limits_{k= 1}^{M}S^{k}\,,
}
\eq
and substitute this into \eqref{fg3} to obtain the Lax matrix of the rational reduced Gaudin model
\beq\label{fg13}
\displaystyle{
	zL_{ij}(z) = v_{i}\delta_{ij}  + \sum\limits_{k = 1}^{M}\frac{z_{k}S^{k}_{ij}}{z-z_{k}}\,.
}
\eq
The resulting Lax matrix \eqref{fg13} satisfies the $r$-matrix relation with respect to brackets \eqref{fg7} \cite{brad}, thus its traces of powers are Poisson commuting functions.
\subsection{$\mathfrak{gl}_{M}$ classical XXX spin chain}
We again start with the direct sum of Poisson-Lie structures on $\mathfrak{gl}_{M}^{*\times N}$:
\beq\label{s1}
\displaystyle{
	\{S^{i}_{ab},S^{j}_{cd}\} = \delta_{ij}(S^{i}_{ad}\delta_{cb} - S^{i}_{cb}\delta_{ad})\,.
}
\eq
The XXX classical spin chain \cite{Heis} is defined by the classical monodromy matrix
\beq\label{s2}
\displaystyle{
	T(v) = Z(I + \frac{S^{N}}{v-v_{N}})...(I + \frac{S^{1}}{v-v_{1}}) \in \mathrm{Mat}_{M}\,,
}
\eq
where $Z = \mathrm{diag}\{z_{1},...,z_{M}\}$\,. For $T(z)$ one has a quadratic $r$-matrix structure
\beq\label{s3}
\displaystyle{
	\{T_{1}(z),T_{2}(w)\} = \Big[\frac{P_{12}}{z-w},T_{1}(z)T_{2}(w)\Big]\,.
}
\eq
Let us assume that all $v_{i}$ are distinct. Then, one may rewrite \eqref{s2} as
\beq\label{s4}
\displaystyle{
	T(v) = Z(I + \sum\limits_{k = 1}^{M}\frac{B^{k}}{v-v_{k}})\,.
}
\eq
The Poisson structure for $B^{k}$ may be computed by residues, since
\beq\label{s5}
\displaystyle{
	B^{k} = \frac{1}{2\pi i} \oint_{\gamma_{i}}\mathrm{d}v\, Z^{-1}T(v) \in \mathrm{Mat}_{M}\,,
}
\eq
where $\gamma_{i}$ is a small contour encircling $v_{i}$. By virtue of the $rTT$ structure \eqref{s3} one obtains
\beq\label{s6}
\displaystyle{
	\{B_{ab}^{i},B_{cd}^{j}\} = (1-\delta_{ij})\frac{B_{cb}^{i}B_{ad}^{j} - B_{cb}^{j}B_{ad}^{i}}{v_{i}-v_{j}} + \delta_{ij}(B^{i}_{ad}\delta_{cb} - B^{i}_{cb}\delta_{ad} - \sum\limits_{k: k\neq i}^{N}\frac{B_{cb}^{i}B_{ad}^{k} - B_{cb}^{k}B_{ad}^{i}}{v_{i}-v_{k}})\,.
}
\eq
For $\mathrm{rank} = 1$ variables $B^{i}_{ab} = \xi^{a}_{i}\eta^{b}_{i}$, this bracket follows from $(\ref{fg8}) - (\ref{fg10})$.
\subsection{Statement of duality}
As one can see, in the $\mathrm{rank} = 1$ case both the rational reduced Gaudin model \eqref{fg13} and the XXX spin chain \eqref{s4} are defined on the same phase space parameterized by $\xi_{a}^{i}, \eta^{b}_{j}$ with Poisson structure $(\ref{fg8} - \ref{fg10})$. With the help of these variables, both models can be written in a similar form. The Gaudin model Lax matrix up to $z$ factor
\beq\label{xx2}
\displaystyle{
	L(z) = V + \xi Z(z-Z)^{1}\eta\, \in \Mat\,,
}
\eq
and spin chain monodromy matrix
\beq\label{xx3}
\displaystyle{
	T(v) = Z(I + \eta (v-V)^{-1}\xi)\, \in \mathrm{Mat}_{M}\,.
}
\eq
It turns out that the spectral curves of both models coincide \cite{MMZZR1}.
\begin{theorem}\label{lx1}
	\beq\label{x13}
	\displaystyle{
		\det(v-V)\det(v-T(z)) = \det(z-Z)\det(z-L(v))\,.
	}
	\eq
\end{theorem}
We omit the proof here, since the infinite-dimensional limit of this duality will reproduce the finite-dimensional computations. For the details of this proof, see \cite{MMZZ, MMZZR1, PZ}.
\section{Large $N$ limit of spectral duality}
\setcounter{equation}{0}
In this section, we describe the $A_{\h}$ reduced Gaudin model, as well as its dual model and the spectral duality statement. 
\subsection{$A_{\h}$ rational reduced Gaudin model}
For the infinite-dimensional case, we construct the rational reduced Gaudin model only for the $"\mathrm{rank} = 1"$ variables. To this end, we consider $2M$ functions on the circle $\xi^{a}$,$\eta^{b} \in L^{1}(S^{1})$:
\beq\label{ag1}
\displaystyle{
	\xi^{a}(\theta) = \sum\limits_{m \in \mZ}\xi^{a}_{m}e^{im\theta}\,,
}
\eq
\beq\label{ag2}
\displaystyle{
	\eta^{b}(\theta) = \sum\limits_{m \in \mZ}\eta^{b}_{m}e^{im\theta}\,.
}
\eq
Let $V(\phi_{1}) \in C^{\infty}(\mathbb{T}^{2})((\h))$ be a function with matrix modes $v_{i}$. This means that in the $\mn$ it is represented by
\beq\label{ag3}
\displaystyle{
	\sum\limits_{i \in \mZ}v_{i}E_{ii}\,,
}
\eq
and 
\beq\label{ag4}
\displaystyle{
	V(\phi_{1}) = \frac{i}{4\pi \h}\sum\limits_{m,k \in \mZ}v_{k}\om^{km}e^{im\phi_{1}}\,.
}
\eq
For this formula to be consistent $|v_{k}|^{2}$ should decay faster than $(1+ k^{2})^{n}$ for any $n$. We demand now that the absolute values of the modes $\xi_{i}^{a},\eta_{j}^{b}$ decay faster than  $|v_{k}|^{2}$, to define the functions 
\beq\label{ag5}
\displaystyle{
	F_{V}[f,g](\vec{\phi}) = \sum\limits_{\vec{m} \in \mZ^{2}: m_{1} \neq m_{2}} \frac{f_{m_{1}}g_{m_{2}}}{v_{m_{1}} - v_{m_{2}}} e^{i\vec{m}\vec{\phi}}\,,
}
\eq
\beq\label{ag6}
\displaystyle{
	G_{V}[f,g](\vec{\phi}) = \sum\limits_{m \in \mZ} \Big(\sum\limits_{p \in \mZ: p\neq m}\frac{f_{p}g_{p}}{v_{m} - v_{p}}\Big)e^{i\vec{m}(\phi_{1}+\phi_{2})} \,.
}
\eq

With the help of these functions, we define
\beq\label{ag7}
\displaystyle{
	\{\xi^{a}(\phi_{1}),\eta^{b}(\phi_{2})\} = - \delta^{ab}\big(\delta(\phi_{1} + \phi_{2}) + G_{V}[\xi^{a},\eta^{b}](\vec{\phi})\big)\,,
}
\eq
\beq\label{ag8}
\displaystyle{
	\{\xi^{a}(\phi_{1}),\xi^{b}(\phi_{2})\} = F_{V}[\xi^{b},\xi^{a}](\vec{\phi})\,,
}
\eq
\beq\label{ag9}
\displaystyle{
	\{\eta^{a}(\phi_{1}),\eta^{b}(\phi_{2})\} = -F_{V}[\eta^{b},\eta^{a}](\vec{\phi})\,.
}
\eq
This means that for the modes $\xi_{i}^{a},\eta_{j}^{b}$ the Poisson bracket coincides with $(\ref{fg8}) - (\ref{fg10})$, so the Jacobi identity is fulfilled.

We follow \cite{PZ1} and define the $\rank = 1$ variables to be equal to
\beq\label{ag10}
\displaystyle{
	A^{[k]}(\vec{\phi}) =  \frac{i}{4\pi\h}\eta^{k}(\phi_{2}) \star \delta(\phi_{1}) \star \xi^{k}(-\phi_{2})\,.
}
\eq
The $\mn$ modes of $A^{k}$ are expressed in the same form as the finite-dimensional "$\mathrm{rank} = 1$" matrix 
\beq\label{ag11}
\displaystyle{
   M^{[k]}_{ij} = \xi^{k}_{i}\eta_{j}^{k}\,,
}
\eq
such that 
\beq\label{ag12}
\displaystyle{
	A^{[k]}_{\vec{m}} = \frac{i}{4\pi\h} \om^{\frac{m_{1}m_{2}}{2}}\sum\limits_{i\in \mZ}\om^{m_{1}i}\xi^{k}_{i}\eta^{k}_{i+m_{2}}\,,
}
\eq
\beq\label{ag13}
\displaystyle{
A^{[k]}(\vec{\phi}) = \sum\limits_{\vec{m} \in \mZ}	A^{[k]}_{\vec{m}}e^{i\vec{m}\vec{\phi}}\,.
}
\eq
Note that the matrix product of $M^{[a]}$ and$M^{[b]}$ is well defined, since
\beq\label{ag14}
\displaystyle{
	 \sum\limits_{k \in \mZ }M^{[a]}_{ik} M^{[b]}_{kj} = \xi^{a}_{i}\eta^{b}_{j}\sum\limits_{k \in \mZ}\eta^{a}_{k}\xi^{b}_{k} = \xi^{a}_{i}\eta^{b}_{j} \frac{1}{2\pi}\int\limits_{S^{1}}\mathrm{d}\theta\, \xi^{b}(\theta)\eta^{a}(-\theta)\,.
}
\eq
The homomorphism \eqref{n5} is still relevant, therefore we have
\beq\label{ag15}
\displaystyle{
A^{[a]} \star A^{[b]} (\vec{\phi}) = \sum\limits_{i,j}W_{\vec{m},ij}\sum\limits_{k \in \mZ }M^{[a]}_{ik} M^{[b]}_{kj}\,.	
}
\eq
The same argument is also applies to the trace:
\beq\label{ag16}
\displaystyle{
	\tr A^{[k]} = \frac{1}{4\pi^{2}}\int_{\mathbb{T}^{2}}\mathrm{d}\vec{\phi}\,A^{[k]}(\vec{\phi}) = \tr M^{[k]}\,.
}
\eq

Finally, the $A_{\h}$ rational reduced Gaudin model is defined by the Lax matrix:
\beq\label{ag}
\displaystyle{
	L(z,\vec{\phi}) = V(\phi_{1}) + \frac{i}{4\pi\h}\sum\limits_{k = 1}^{M}\frac{z_{k}\eta^{k}(\phi_{2}) \star \delta(\phi_{1}) \star \xi^{k}(-\phi_{2})}{z-z_{k}}\,.
}
\eq
In contrast with the finite-dimensional case, the integrability of the reduced model does not follow from the definition. We will prove its integrability below by the spectral duality statement.
\subsection{Dual model}
To construct the dual side, we define new variables 
\beq\label{dm1}
\displaystyle{
	\tilde{A}^{[k]}_{ab} = \frac{1}{4\pi^{2}}\int\limits_{\mathbb{T}^{2}}\mathrm{d}\, \vec{\phi}(V(\phi_{1}))^{k} \star \frac{i}{4\pi\h}\eta^{b}(\phi_{2})\star\delta(\phi_{1})\star \xi^{a}(-\phi_{2})\,, \quad a,b = 1,...,M\,,
}
\eq
or in the matrix modes
\beq\label{dm2}
\displaystyle{
	\tilde{A}^{[k]}_{ab} = \sum\limits_{i \in \mZ}v_{i}^{k}\xi^{a}_{i}\eta^{b}_{i}\,.	
}
\eq
The role of conserved quantities generator will be played 
\beq\label{dx1}
\displaystyle{
	\tilde{T}(v) = Z\Big(I+ \sum\limits_{k = 1}^{\infty}\frac{\tilde{A}^{[k-1]}}{v^{k}}\Big)\, \in \mathrm{Mat}_{M}[[v^{-1}]]\,.
}
\eq
The following proposition shows that \eqref{dx1} is an infinite-dimensional generalization of the XXX spin chain \eqref{s4}.
\begin{lemma}
	\beq\label{dx2}
	\displaystyle{
		\{\tilde{T}_{1}(z),\tilde{T}_{2}(w)\} = \Big[\frac{P_{12}}{z-w},\tilde{T}_{1}(z)\tilde{T}_{2}(w)\Big]\,.	
	}
	\eq
\end{lemma}
\begin{proof}
	Comparing coefficients before powers of  $z^{-1}$ and $w^{-1}$ we obtain two sets of relations
	\beq\label{p4}
	\displaystyle{
		\{\tilde{A}_{1}^{[0]},\tilde{A}^{[k-1]}_{2}\} = \Big[P_{12},\tilde{A}_{2}^{[k-1]}\Big]\,,\quad k \in \mathbb{N}\,,
	}
	\eq
	\beq\label{p5}
	\displaystyle{
		\{\tilde{A}_{1}^{[k_{1}]},\tilde{A}^{[k_{2}-1]}_{2}\} - \{\tilde{A}_{1}^{[k_{1}-1]},\tilde{A}^{[k_{2}]}_{2}\} = \Big[P_{12},\tilde{A}_{1}^{[k_{1}-1]}\tilde{A}_{2}^{[k_{2}-1]}\Big]\,,\quad k_{1},k_{2} \in \mathbb{N}\,.
	}
	\eq
	Our strategy is to use the bracket \eqref{s6}, which is equivalently written as

	\beq\label{p2}
	\displaystyle{
		\{B^{i}_{1},B^{j}_{2}\} =  \frac{1-\delta^{ij}}{v_{i}-v_{j}}\Big[P_{12},B^{i}_{1}B^{j}_{2}\Big] + \delta^{ij}\Big(\Big[P_{12},B_{2}\Big] - \sum\limits_{k \neq i }\frac{1}{v_{i}-v_{k}}\Big[P_{12},B^{i}_{1}B^{k}_{2}\Big]\Big)\,.
	}
	\eq
	and prove that $(\ref{p4})-(\ref{p5})$ hold. We combine it with the definition \eqref{dm2} to conclude that 
	\beq\label{p3}
	\displaystyle{
		\{\tilde{A}^{[0]}_{1},\tilde{A}^{[k]}_{2}\} = \Big[P_{12},\tilde{A}^{[k]}_{2}\Big] \,,\quad k \in \mathbb{N}\,,
	}
	\eq
	\beq\label{pp3}
	\displaystyle{
		\{\tilde{A}^{[k_{1}]}_{1},\tilde{A}^{[k_{2}]}_{2}\} = \Big[P_{12},\tilde{A}^{[k_{1}+k_{2}]}_{2}\Big] - \sum\limits_{l = 0}^{k_{2}-1}\Big[P_{12},\tilde{A}^{[k_{1}+k_{2}-1-l]}_{1}\tilde{A}^{[l]}_{2}\Big]\,, \quad k_{1},k_{2} \in \mathbb{N}\,.
	}
	\eq
	The relation \eqref{p4} obviously holds, then we obtain
	\beq\label{p6}
	\begin{array}{c}
		\displaystyle{
			\{\tilde{A}_{1}^{[k_{1}]},\tilde{A}^{[k_{2}-1]}_{2}\} - \{\tilde{A}_{1}^{[k_{1}-1]},\tilde{A}^{[k_{2}]}_{2}\} = \sum\limits_{l = 0}^{k_{2}-1}\Big[P_{12}\tilde{A}_{1}^{[k_{1}+k_{2}-2-l]}\tilde{A}_{2}^{[l]}\Big] - \sum\limits_{l = 0}^{k_{2}-2}\Big[P_{12}\tilde{A}_{1}^{[k_{1}+k_{2}-2-l]}\tilde{A}_{2}^{[l]}\Big]  = 
		}
		\\ \ \\ 
		\displaystyle{
			= \Big[P_{12},\tilde{A}_{1}^{[k_{1}-1]}\tilde{A}_{2}^{[k_{2}-1]}\Big]\,,
		}
	\end{array}
	\eq
	which proves \eqref{dx2}.
\end{proof}

\subsection{Statement of duality}
For the $A_{\h}$ model, we define the spectral power series to be
\beq\label{x14}
\displaystyle{
	\Gamma_{\infty}(v,z) = \exp \tr \ln_{\star}\Big(1 - \frac{1}{v}L(z,\vec{\phi})\Big)\,,
}
\eq
with
\beq\label{x15}
\displaystyle{
	\ln_{\star}\Big(1 - \frac{1}{v}L(z,\vec{\phi})\Big) = -\sum\limits_{k = 1}^{\infty}\frac{\tr L^{\star k}(z)}{k v^{k}}\,,
}
\eq
such that $\Gamma_{\infty}(v,z)$ is a formal power series in $v^{-1}$. Its coefficients are functions of $z$ and generate involutive Hamiltonians. The spectral duality statement should relate $\Gamma_{\infty}(v,z)$  of the $A_{\h}$ model \eqref{ag} with the spectral curve of the dual model \eqref{dx1}. The following theorem provides an infinite-dimensional analogue of the determinant relation \eqref{x13}.
\begin{theorem}\label{lsg2}
	\beq\label{sg11}
	\displaystyle{
		\exp\Big(\tr \ln_{\star}\Big(I-\frac{V(\phi_{1})}{v}\Big)\Big)\det(z-\tilde{T}(v)) = \det(z-Z)	\exp\Big(\tr \ln_{\star}\Big(I-\frac{1}{v}L(z,\vec{\phi})\Big)\Big)\,.
	}
	\eq

\end{theorem}
\begin{proof}
We follow \cite{PZ1} and carry out the proof in the matrix modes. Let us first note that 
\beq\label{psg1}
\displaystyle{
	\det(I - D) = \exp\tr\ln(1-D)\,.
}
\eq
Thus, $\det(I - D)$ may be expanded into powers of $\tr(D^{k})$, whose coefficients do not depend on the shape of $D$. If we assume that $D$ is an infinite-dimensional matrix, this power series will coincide with \eqref{x14}. 
We apply this fact to 
\beq\label{psg2}
\displaystyle{
	\det(z-Z)^{-1}\det(z-\tilde{T}(v)) = \det\Big( I - (z-Z)^{-1}Z\sum\limits_{k = 1}^{\infty}\frac{\tilde{A}^{k-1}}{v^{k}}\Big).
}
\eq
Then
	\beq\label{sg12}
	\begin{array}{c}
		\displaystyle{
			\tr \Big((z-Z)^{-1}Z\sum\limits_{k = 1}^{\infty}\frac{\tilde{A}^{k-1}}{v^{k}}\Big)^{n} =
		}
		\\ \ \\
		\displaystyle{
			=
			\sum\limits_{k_{1},...,k_{n} = 1}^{\infty}\frac{1}{v^{k_{1}+...+k_{n}}}\sum\limits_{p_{1},...,p_{n} \in \mZ}
			\sum\limits_{a_{1},...,a_{n} = 1}^{M}
			\frac{z_{a_{1}}v_{p_{1}}^{k_{1}-1}\xi^{a_{1}}_{p_{1}}\eta^{a_{2}}_{p_{1}}}{z-z_{a_{1}}}...
			\frac{z_{a_{n}}v_{p_{n}}^{k_{n}-1}\xi^{a_{n}}_{p_{n}}\eta^{a_{1}}_{p_{n}}}{z-z_{a_{n}}} =
		}
		\\ \ \\
		\displaystyle{
			=\sum\limits_{k_{1},...,k_{n} = 1}^{\infty}\frac{1}{v^{k_{1}+...+k_{n}}}\sum\limits_{p_{1},...,p_{n} \in \mZ}\sum\limits_{a_{1},...,a_{n} = 1}^{M} \frac{z_{a_{1}}v_{p_{1}}^{k_{1}-1}\xi^{a_{1}}_{p_{1}}\eta^{a_{1}}_{p_{n}}}{z-z_{a_{1}}}...\frac{z_{a_{2}}v_{p_{2}}^{k_{2}-1}\xi^{a_{2}}_{p_{2}}\eta^{a_{2}}_{p_{1}}}{z-z_{a_{2}}}\,.}
	\end{array}
	\eq
	The last line corresponds to the trace of the product of the infinite matrices. We express them as functions
	\beq\label{s121}
	\displaystyle{
		\tr \Big((z-Z)^{-1}Z\sum\limits_{k = 1}^{\infty}\frac{\tilde{A}^{k-1}}{v^{k}}\Big)^{n}=\tr\Big(\sum\limits_{k =1}^{\infty}\frac{(V(\phi_{1}))^{k-1}}{v^{k}}\star \frac{i}{4 \pi \h}\sum\limits_{a = 1}^{M}\frac{z_{a}\eta^{a}(\phi_{2})\star \delta(\phi_{1}) \star \xi^{a}(-\phi_{2}) }{z-z_{a}}\Big)^{\star n}\,.}	
	\eq
This means that

	\beq\label{sg13}
	\begin{array}{c}
		\displaystyle{
			\exp\Big(\tr \ln_{\star}\Big(I-\frac{V(\phi_{1})}{v}\Big)\Big)\det(z-\tilde{T}(v)) = \det(z-Z) \exp\Big(\tr \ln_{\star}\Big(I-\frac{V(\phi_{1})}{v}\Big)+ 
		}
		\\ \ \\
		\displaystyle{
			+\tr \ln_{\star}\Big(I-\sum\limits_{k =1}^{\infty}\frac{(V(\phi_{1}))^{k-1}}{v^{k}}\star \frac{i}{4 \pi \h}\sum\limits_{a = 1}^{M}\frac{z_{a}\eta^{a}(\phi_{2})\star \delta(\phi_{1}) \star \xi^{a}(-\phi_{2}) }{z-z_{a}}\Big)\Big)	\,.
		}
	\end{array}
	\eq
	Finally, we use the identity
	\beq\label{rsd22}
	\displaystyle{
		\tr\ln_{\star}(I - A)\star(I - B) = \tr\ln_{\star}(I - A)+ \tr\ln_{\star}(I - B)\,,
	}
	\eq
    to rewrite the expressions inside the exponent  
	\beq\label{sg14}
	\begin{array}{c}
		\displaystyle{
			\tr \ln_{\star}\Big(I-\frac{V(\phi_{1})}{v}\Big)
			+\tr \ln_{\star}\Big(I-\sum\limits_{k =1}^{\infty}\frac{(V(\phi_{1}))^{k-1}}{v^{k}}\star \frac{i}{4 \pi \h}\sum\limits_{a = 1}^{M}\frac{z_{a}\eta^{a}(\phi_{2})\star \delta(\phi_{1}) \star \xi^{a}(-\phi_{2}) }{z-z_{a}}\Big) = 
		}
		\\ \ \\
		\displaystyle{
			=  \tr \ln_{\star}\Big(\Big[I-\frac{V(\phi_{1})}{v}\Big]\star\Big[I-\sum\limits_{k =1}^{\infty}\frac{(V(\phi_{1}))^{k-1}}{v^{k}}\star \frac{i}{4 \pi \h}\sum\limits_{a = 1}^{M}\frac{z_{a}\eta^{a}(\phi_{2})\star \delta(\phi_{1}) \star \xi^{a}(-\phi_{2}) }{z-z_{a}}\Big]\Big) = 
		}
		
		\\ \ \\
		\displaystyle{
			=  I-\frac{1}{v}L(z,\vec{\phi})	\,,
		}
		
	\end{array}
	\eq
	which finishes the proof.
\end{proof}\\

As a corollary of this theorem, we have the integrability  of  the $A_{\h}$ Gaudin model \eqref{ag}. Indeed, the coefficients of the spectral power series \eqref{x14} Poisson commute with respect to the brackets $(\ref{ag7}) - (\ref{ag9})$, since they are expressed via coefficients of the monodromy matrix \eqref{dx1} spectral curve. 

Note that the infinite-dimensional limit of the spin chain \eqref{dx1} has the same analytic structure and Poisson brackets as those expected for a twisted 1+1 field theory monodromy matrix \cite{FT}. There is no direct evidence that the variables \ref{dm1} can be expressed as integrals of 1+1 fields (as the residues of the monodromy matrix should be). However, we hope that in future studies this theory will be directly connected with some 1+1 soliton equation.

\subsection*{Acknowledgments}
This work was performed at the Steklov International Mathematical Center and supported by the
Ministry of Science and Higher Education of the Russian Federation (agreement no. 075-15-2025-303). The author would like to thank Andrei Zotov for fruitful discussions and ideas.

\begin{small}

\end{small}

\end{document}